\documentclass[12pt]{article}
\usepackage{amsfonts}
\usepackage{amsmath}
\usepackage{epsfig}
\usepackage{latexsym}
\usepackage{amssymb}
\usepackage{graphicx}
\usepackage[usenames,dvipsnames]{color}

\newtheorem{theorem}{Theorem}

\newtheorem{proposition}[theorem]{Proposition}

\newenvironment{proof}[1][Proof]{\noindent\textbf{#1.} }{\ \rule{0.5em}{0.5em}}

\begin{document}

\title{ A solution of the Gaussian optimizer conjecture}
\author{V. Giovannetti$^{1}$, A. S. Holevo$^{2}$, R. Garc{\'i}a-Patr\'on$
^{3,4}$ \\
\\
{\small $^{1}$NEST, Scuola Normale Superiore and Istituto Nanoscienze-CNR, }
\\
{\small I-56127 Pisa, Italy,} \\
{\small $^2$Steklov Mathematical Institute, RAS, Moscow, Russia}\\
{\small $^{3}$ Quantum Information and Communication, Ecole Polytechnique de
Bruxelles,} \\
{\small CP 165, Universite Libre de Bruxelles, 1050 Bruxelles, Belgium} \\
{\small $^{4}$Max-Planck Institut f\"{u}r Quantenoptik, }\\
{\small Hans-Kopfermann Str. 1, D-85748 Garching, Germany} }
\date{}
\maketitle

\begin{abstract}
The long-standing conjectures of the optimality of Gaussian inputs for
Gaussian channel and Gaussian additivity are solved for a broad class of
covariant or contravariant Bosonic Gaussian channels (which includes in
particular thermal, additive classical noise, and amplifier channels)  restricting to the class of states with finite second moments. We
show that the vacuum is the input state which minimizes the entropy at the
output of such channels. This allows us to show also that the classical
capacity of these channels (under the input energy constraint) is additive
and is achieved by Gaussian encodings.
\end{abstract}

\section{Introduction}

The capacity of a communication channel is the maximum rate at which
information (measured in bits per channel use) can be transmitted from
sender to receiver with asymptotically vanishing error~\cite{COVER}. It
provides an operationally defined measure of the communication efficiency of
the channel by setting the ultimate limit at which messages can be
transferred reliably. An explicit formula for this quantity in terms of an
entropic functional of the conditional probability distribution that defines
the noise model is given by the \textit{Shannon noisy-channel coding theorem}
~\cite{Shannon}. For channels with additive Gaussian noise it amounts to the
famous formula
\begin{equation}
C=\frac{1}{2}\log \left( 1+E/N\right) ,  \label{shan}
\end{equation}
where $E/N$ is signal-to-noise ratio.

In the context of quantum information theory~\cite{BENSHOR,HG},
communication channels are described by linear, completely positive, trace
preserving (CPTP) maps $\Phi $ which transform input density matrices to
their output counterparts. For these models the analog of the Shannon noisy
coding theorem was obtained in Refs.~\cite{HOLEVO98,SCHWEST} for
finite-dimensional channels. It says that that the associated (classical) capacity is
equal to
\begin{equation}
C(\Phi )=\lim_{n\rightarrow \infty }\ \frac{1}{n}C_{\chi }(\Phi ^{\otimes
n}),  \label{CfromC1}
\end{equation}
where $\Phi ^{\otimes n}$ is the map describing $n$ channel uses (memoryless
noise model), while $C_{\chi }$ is the $\chi $-capacity defined by the
expression
\begin{equation}
C_{\chi }(\Psi )=\max_{\mathcal{E}}\left\{ S(\Psi \lbrack \sum_{j}p_{j}\rho
_{j}])-\sum_{j}p_{j}S(\Psi \lbrack \rho _{j}])\right\} ,
\label{oneshotunassist}
\end{equation}
where the maximization is performed over the set of (possibly constrained)
input ensembles $\mathcal{E}=\{p_{j};\rho _{j}\}$ ($p_{j}$ being
probabilities, while $\rho _{j}$ being density matrices) and $S(\rho )= -
\mbox{Tr}\rho \log \rho $ is the von Neumann entropy.

A special class of maps which play a fundamental role in quantum information
theory constitute the so called Bosonic Gaussian Channels (BGCs) \cite{HOWE}.
Among other noise models they describe thermal, attenuation, and
amplification processes for all those communication setups where messages
are encoded into the modes of the electromagnetic field (say optical
fibers), i.e. the most common quantum communication architectures~\cite
{CAVES}. Computing the capacity of these channels is hence an important
problem which has profound implications both from theoretical and technical
point of view. Since BGCs act in infinite dimensions, the corresponding
generalization of the coding theorem (\ref{CfromC1}), (\ref{oneshotunassist}),
allowing for input constraint and continuous state ensembles is required, see \cite
{h2,HOSHI} and also \cite{h}, Ch.11. A long-standing conjecture, first
proposed in~\cite{HOWE}, is that for these maps and energy constraints
Gaussian encodings should provide the optimal communication rates, allowing
to restrict the maximization of Eq.~(\ref{oneshotunassist}) over the set of
Gaussian ensembles, a task which can be performed analytically (\textit{
optimal Gaussian ensemble conjecture}); moreover $C_{\chi }(\Phi )$ should
be additive for BGC (\textit{Gaussian additivity conjecture}) ensuring that
no limit in (\ref{CfromC1}) is necessary so that $C(\Phi )=C_{\chi }(\Phi )$.
It turns out that in some important situations such statements can be
reduced to similar conjectures on the output entropy of the channel~$\Phi $~
\cite{conj1,conj2,conj3,gp}. In this formulation, Gaussian input states
are supposed to provide the minimum for this quantity~(\textit{minimum
output entropy conjecture}). Despite a number of indirect evidences of
correctness (see e.g. Refs.~\cite{konig1,konig2,natphot}), up to date both
the optimal Gaussian ensemble and the minimum output entropy conjectures
remained open, except for a special class of quantum-limited attenuator (or
lossy) channels~\cite{gio}. Extending these results to a broader set of BGCs
proved to be one of the major challenges in quantum information theory. In
the present paper we give a solution to these problems by showing that
vacuum input minimizes the entropy at the output of any multimode
gauge-covariant or contravariant BGC, while the capacity constrained with an
oscillator energy operator is attained on the corresponding Gaussian
ensemble of coherent states. In our solution we restrict all
optimizations to the class of states with finite second moments which natural when dealing with capacities
of channels with energy-constrained inputs. As for the Gaussian minimal output entropy problem, this restriction can be relaxed but we postpone the solution to the future
work \cite{ANDREA}.

The manuscript is organized as follows. In Sec.~\ref{sec1} we introduce the
notation and define important classes of covariant quantum-limited
attenuators, quantum-limited amplifiers, and quantum-limited contravariant
channels. Furthermore we show (Proposition \ref{prop1}) that any Gaussian
gauge-covariant channel can be expressed as a concatenation of a
quantum-limited attenuator followed (in the Schr\"{o}dinger picture) by a
quantum-limited amplifier. In Sec.~\ref{sec:complandadd} we recall the
additivity properties of entanglement-breaking channels, and observe that
the quantum-limited contravariant channel is entanglement-breaking
(Proposition \ref{prop3}) and shares the additivity properties with the
complementary quantum-limited covariant amplifier \ (Proposition \ref{prop2}
). In Sec. \ref{sec:reduction} we show that proving the minimal output
entropy conjecture for an arbitrary covariant or contravariant channel
reduces to proving it for a single-mode quantum-limited amplifier
(Proposition \ref{prop4}). Also we show that proving the last fact allows
one to compute the classical capacity of an arbitrary covariant
(contravariant) channel under the energy constraint. In Sec.~\ref{sec:new}
we derive key identities that lead to the proof of the conjecture presented
in Sec.~\ref{sec:6}. Specifically, Sec.~\ref{sec:new} is devoted to
characterization of the output states of the quantum-limited contravariant
channel. In Proposition~\ref{prop5} we give an explicit measure-reprepare
decomposition of these maps and present (Proposition~\ref{prop6}) a
covariant amplifier channel whose output states have the same spectrum as
the original channel. Building up on these elements, in Sec.~\ref{sec:6} we
establish an important decomposition (Proposition~\ref{propo10}) from which
one can finally deduce that the minimal output entropy of a single-mode
covariant amplifier is indeed achieved by the vacuum. In view of the results
of Sec.~\ref{sec:reduction}, this proves both the Gaussian encoding and the
minimal output entropy conjectures for the whole class of covariant and
contravariant channels. The paper is concluded with Sec.~\ref{sec:conc}
where we present specific examples which can be considered as counterparts
of the Shannon formula (\ref{shan}) for quantum channels, and discuss further
implications of our findings.

\section{Gaussian gauge-covariant and contravariant channels}

\label{sec1}

The scenery for considering gauge-covariant states and channels is $s-$
dimensional complex Hilbert space $\mathbf{Z}$ which can be considered as $%
2s-$dimensional real space equipped with the symplectic form $\Delta
(z,z^{\prime })=\Im z^{\ast }z^{\prime }.$ To be specific, we consider
vectors in $\mathbf{Z}$ as $s-$dimensional complex column vectors, in which
case (complex-linear) operators in $\mathbf{Z}$ are represented by complex $
s\times s-$matrices, and $^{\ast }$ denotes Hermitian conjugation. The gauge
group acts in $\mathbf{Z}$ as multiplication by $e^{i\phi },$ where $\phi $
is real number called phase. The Weyl quantization is described by
displacement operators $D(z)$ acting irreducibly in the representation space
$\mathcal{H}$ and satisfying the relation
\begin{equation}
D(z)D(z^{\prime })=\exp \left( -i\Im \bar{z} z^{\prime }\right)
D(z+z^{\prime }). \label{defDD}
\end{equation}
Introducing the annihilation (resp. creation)
operators of the system $a_{j}$ and $a_{j}^{\dag }$  which satisfy the commutation relations $\left[
a_{j\,,}a_{k}^{\dag }\right] =\delta _{jk}I$, we recall that $D(z)$ can be
expressed as
\begin{equation}  \label{displacement}
D(z)=\exp[ \mathbf{a}^{\dag }z - z^*\mathbf{a}]=\exp \sum_{j=1}^{s}\left(
z_{j}a_{j}^{\dag }-\bar{z}_{j}a_{j}\right),
\end{equation}%
where $\mathbf{a}=\left[ a_{1},\dots ,a_{s}\right] ^{t}$ and $\mathbf{a}%
^{\dag }=\left[ a_{1}^{\dag },\dots ,a_{s}^{\dag }\right]$ are respectively
column and row vectors. Next let $\Lambda $ be the antiunitary operator of
complex conjugation in $ \mathbf{Z}$ which anticommutes with multiplication
by $i$ and satisfies $ \Lambda ^{\ast }=\Lambda ,\,\Lambda ^{2}=I.$ The
associated transposition map $\mathrm{T}$ acting on operators in $\mathcal{H}
$ can then be defined by the relation $\mathrm{T}[D(z)]=D(-\Lambda z),\,z\in
\mathbf{Z}$, or equivalently $\mathrm{T} [a_{j\,}]=a_{j}^{\dag }$ and $%
\mathrm{T}[\mathbf{a}]=\left[ a_{1}^{\dag },\dots
,a_{s}^{\dag }\right] ^{t}$ (notice that the last is a column vector
different from $\mathbf{a}^{\dag}$).

The gauge group has the unitary representation $\phi \rightarrow U_{\phi
}=e^{i\phi N}$ in $\mathcal{H}$ where $N=\sum_{j=1}^{s}a_{j}^{\dag }a_{j}$
is the total number operator. A state $\rho $ is then said to be
gauge-invariant if it commutes with all $U_{\phi }$. or, equivalently, if
its (symmetrically ordered) characteristic function $\mathcal{F}(z)=\text{Tr}%
\rho D(z)$~\cite{walls} is invariant under the action of the gauge group.
In particular Gaussian gauge-invariant states are described by the property
\begin{equation}\label{gausstate}
\mathrm{Tr}\rho D(z)=\exp \left( -z^{\ast }\alpha z\right) ,
\end{equation}
where $\alpha$ is a complex-linear covariance operator satisfying $\alpha
\geq I/2$. The vacuum state $\rho _{vac}=|0\rangle \langle 0|$ is an element
of this set with $\alpha =I/2$.

A channel $\Phi $ with the input space $\mathcal{H}_{A}$ and the output
space $\mathcal{H}_{B}$ satisfying
\begin{equation*}
\Phi \lbrack e^{i\phi N_{A}}\rho e^{-i\phi N_{A}}]=e^{\pm i\phi N_{B}}\Phi
\lbrack \rho ]e^{\mp i\phi N_{B}},
\end{equation*}%
is called \emph{gauge-covariant}\ (\emph{gauge-contravariant}). We denote by
$s_{A}=\dim \mathbf{Z}_{A}$, $s_{B}=\dim \mathbf{Z}_{B}$ the numbers of
modes of the input and output of the channel.

In the Heisenberg representation, a multimode bosonic Gaussian
gauge-covariant channel $\Phi $ \cite{htw,h} is described by the action of
its adjoint $\Phi ^{\ast }$ onto displacement operators as follows:
\begin{equation}
\Phi ^{\ast }[D_{B}(z)]=D_{A}(K^{\ast }z)\exp \left( -z^{\ast }\mu z\right),
\label{defprima}
\end{equation}%
where $K$ is complex-linear operator from $\mathbf{Z}_{A}$ to $\mathbf{Z}%
_{B} $ and $\mu $ is complex Hermitian operator in $\mathbf{Z}_{B}$
satisfying the inequality (cf. \cite{htw}, eq. (24))
\begin{equation}
\mu \geq \pm \frac{1}{2}\left( I-KK^{\ast }\right) .  \label{ineq}
\end{equation}
Notice, that if (and only if) the operators $K$ and $\mu $ can be
simultaneously diagonalized in an orthonormal basis, then the channel
decomposes into tensor product of one-mode channels (cf. \cite{htw}). In
this case we call the channel \emph{diagonalizable}. The gauge-covariant
channel is \emph{quantum-limited} if $\mu $ is a minimal solution of this
inequality.

Special cases of the maps~(\ref{defprima}) are provided by the attenuators
and amplifier channels, characterized by matrix $K$ fulfilling the
inequalities, $KK^{\ast }\leq I$ and $KK^{\ast }\geq I$ respectively. We are
particularly interested in \emph{quantum-limited attenuator} which
corresponds to
\begin{equation}
KK^{\ast }\leq I,\qquad \qquad \mu =\frac{1}{2}\left( I-KK^{\ast }\right),
\label{mindef1}
\end{equation}
and \emph{quantum-limited amplifier}
\begin{equation}
KK^{\ast }\geq I,\qquad \qquad \mu =\frac{1}{2}\left( KK^{\ast }-I\right).
\label{mindef2}
\end{equation}
These channels are diagonalizable: by using singular value decomposition $%
K=V_{B}K_{c}V_{A},$ where $V_{A},V_{B}$ are unitaries and $K_{c}$ is
(rectangular) diagonal matrix with nonnegative values on the diagonal, we
have $KK^{\ast }=V_{B}K_{c}K_{c}{}^{\ast }V_{B}^{\ast },$ and%
\begin{equation}
\Phi \lbrack \rho ]=U_{B}\Phi _{c}[U_{A}^{\ast }\rho U_{A}]U_{B}^{\ast },
\label{unieq}
\end{equation}
where $\Phi _{c}$ is a tensor product of one-mode a (quantum limited)
channels defined by the matrix $K_c$ and where $U_{A}$, $U_{B}$ are passive
canonical unitary transformations acting on $\mathcal{H}_{A}$ and $\mathcal{H%
}_{B}$ respectively, such that
\begin{equation*}
U_{B}^{\ast }\mathbf{a}U_{B}=V_{A}\mathbf{a},\qquad \qquad U_{A}^{\ast }%
\mathbf{a} U_{A}=V_{B}{}^{t}\mathbf{a},
\end{equation*}
with $\mathbf{a}$
 being the column
vector formed by the annihilation operators introduced in Eq.~(\ref{defDD})  (notice that in particular this
implies $U_{A}|0\rangle =|0\rangle ,U_{B}|0\rangle =|0\rangle$).

A multimode \emph{\ Gaussian gauge-contravariant channel} involving
\textquotedblleft phase inversion\textquotedblright\ acts as
\begin{equation}
\Phi ^{\ast }[D_{B}(z)]=D_{A}(\Lambda K^{\ast }z)\exp \left( -z^{\ast }\mu
z\right) =D_{A}(K^{t}\bar{z})\exp \left( -z^{\ast }\mu z\right) ,
\label{CONTRAV}
\end{equation}%
and $\mu $ is the Hermitian operator satisfying
\begin{equation}
\mu \geq \frac{1}{2}\left( I+KK^{\ast }\right) ,  \label{ineq2}
\end{equation}%
(in the second identity of Eq.~(\ref{CONTRAV}) we used the fact that for all
$z$ one has $\Lambda K^{\ast }z=K^{t}\bar{z}$ with $K^{t}$ being the
transpose of the matrix $K$ and $\bar{z}$ being the column vector obtained
by taking the complex conjugate of the elements of $z$). If the operators $%
K\Lambda $ and $\mu $ can be simultaneously diagonalized in an orthonormal
basis, then this channel is equivalent in the sense of (\ref{unieq}) to tensor
product of one-mode channels and is called \emph{diagonalizable}. These maps
are quantum-limited if
\begin{equation}\label{mindef5}
\mu =\frac{1}{2}\left( I+KK^{\ast }\right)
\end{equation}
and these are diagonalizable similarly to quantum-limited amplifiers.

All the examples of quantum-limited BGSs introduced above can also be
expressed in terms of input-output equations for the column vectors of
annihilation operators as
\begin{eqnarray*}
\mathbf{a}_{B} &=&K\mathbf{a}_{A}+\sqrt{I-KK^{\ast }}\; \mathbf{a}_{E}, \quad 
\mbox{{\it (attenuator)},} \\
\mathbf{a}_{B} &=&K\mathbf{a}_{A}+\sqrt{KK^{\ast }-I}\;  \mbox{T}[\mathbf{a}_{E}],
\quad \mbox{{\it (amplifier)},} \\
\mathbf{a}_{B} &=&- K\; \mbox{T}[\mathbf{a}_{A}]+\sqrt{I+KK^{\ast }}\; \mathbf{a}_{E},
\quad \mbox{{\it (gauge-contravariant)},}
\end{eqnarray*}
where the environment modes associated with $\mathbf{a}_{E}$ are in the
vacuum state.

The concatenation $\Phi =\Phi _{2}\circ \Phi _{1}$ of two channels $\Phi _{1}
$ and $ \Phi _{2}$ obeys the rule
\begin{eqnarray}
K &=&K_{2}K_{1},\quad  \label{c1} \\
\mu &=&K_{2}\mu _{1}K_{2}^{\ast }+\mu _{2}.  \label{c2}
\end{eqnarray}
The following proposition generalizes to many modes the decomposition of
one-mode channels the usefulness of which was emphasized and exploited in the
paper \cite{gp} (see also \cite{cgh} on concatenations of one-mode channels):

\begin{proposition}
\label{prop1} Any bosonic Gaussian gauge-covariant channel $\Phi $ is a
concatenation of quantum-limited attenuator $\Phi _{1}$ and (diagonalizable)
quantum-limited amplifier $\Phi _{2}$.

Any bosonic Gaussian gauge-contravariant channel $\Phi $ is a concatenation
of quantum-limited attenuator $\Phi _{1}$ and (diagonalizable)
quantum-limited gauge-contravariant channel $\Phi _{2}$.
\end{proposition}

\begin{proof}
By inserting
\begin{equation*}
\mu _{1}=\frac{1}{2}\left( I-K_{1}K_{1}^{\ast }\right) =\frac{1}{2}\left(
I-|K_{1}^{\ast }|^{2}\right) ,\quad \mu _{2}=\frac{1}{2}\left(
K_{2}K_{2}^{\ast }-I\right) =\frac{1}{2}\left( \left\vert K_{2}^{\ast
}\right\vert ^{2}-I\right)
\end{equation*}%
into (\ref{c2}) and using (\ref{c1}) we obtain
\begin{equation}
\left\vert K_{2}^{\ast }\right\vert ^{2}=K_{2}K_{2}^{\ast }=\mu +\frac{1}{2}%
(KK^{\ast }+I)\geq \left\{
\begin{array}{c}
I \\
KK^{\ast }%
\end{array}%
\right.  \label{ineq1}
\end{equation}%
from the inequality (\ref{ineq}). By using operator monotonicity of the
square root, we have
\begin{equation*}
\left\vert K_{2}^{\ast }\right\vert \geq I,\quad \left\vert K_{2}^{\ast
}\right\vert \geq \left\vert K^{\ast }\right\vert .
\end{equation*}%
The first inequality (\ref{ineq1}) implies that choosing $K_{2}=\left\vert
K_{2}^{\ast }\right\vert =\sqrt{\mu +\frac{1}{2}(KK^{\ast }+I)}$ and the
corresponding $\mu _{2}=\frac{1}{2}\left( \left\vert K_{2}^{\ast
}\right\vert ^{2}-I\right) ,$ we obtain diagonalizable quantum-limited
amplifier, since $K_{2}$ and $\mu _{2}$ are commuting Hermitian operators.
Notice, that in general we can take $K_{2}=\left\vert K_{2}^{\ast
}\right\vert V,$ where $V$ is arbitrary co-isometry, $VV^{\ast }=I$.

Then with $K_{1}=\left\vert K_{2}^{\ast }\right\vert ^{-1}K$ we obtain,
taking into account the second inequality in (\ref{ineq1})
\begin{equation*}
K_{1}^{\ast }K_{1}=K^{\ast }\left\vert K_{2}^{\ast }\right\vert
^{-2}K=K^{\ast }\left[ \mu +\frac{1}{2}(KK^{\ast }+I)\right] ^{-1}K\leq
K^{\ast }(KK^{\ast })^{-}K\leq I,
\end{equation*}
where $^{-}$ means generalized inverse, which implies $K_{1}^{\ast
}K_{1}\leq I,$ hence $K_{1}$ with the corresponding $\mu _{1}=\frac{1}{2}
\left( I-K_{1}K_{1}^{\ast }\right) $ give the quantum-limited attenuator.

In the case of contravariant channel the equations (\ref{c1}) is replaced
with
\begin{equation*}
K\Lambda =K_{2}\Lambda K_{1}.
\end{equation*}
By substituting this and
\begin{equation*}
\mu _{1}=\frac{1}{2}\left( I-K_{1}K_{1}^{\ast }\right) ,\quad \mu _{2}=\frac{
1}{2}\left( K_{2}K_{2}^{\ast }+I\right)
\end{equation*}
into (\ref{c2}) and using (\ref{ineq2}) we obtain
\begin{equation*}
\left\vert K_{2}^{\ast }\right\vert ^{2}=K_{2}K_{2}^{\ast }=\mu +\frac{1}{2}
(KK^{\ast }-I)\geq KK^{\ast }.
\end{equation*}
Taking $K_{2}=\left\vert K_{2}^{\ast }\right\vert ,$ $\mu _{2}=\frac{1}{2}
\left( |K_{2}^{\ast }|^{2}+I\right) $ gives diagonalizable quantum-limited
gauge-contravariant channel. With $K_{1}=\Lambda ^{-1}\left\vert K_{2}^{\ast
}\right\vert ^{-}K\Lambda =\Lambda \left\vert K_{2}^{\ast }\right\vert
^{-}K\Lambda $ we obtain
\begin{eqnarray}
K_{1}^{\ast }K_{1}&=&\Lambda K^{\ast }\left( \left\vert K_{2}^{\ast
}\right\vert ^{-}\right) ^{2}K\Lambda  \notag \\
&=&\Lambda K^{\ast }\left[ \mu +\frac{1}{ 2}(KK^{\ast }-I)\right]
^{-}K\Lambda \leq K^{\ast }(KK^{\ast })^{-}K\leq I,
\end{eqnarray}
which implies $K_{1}K_{1}^{\ast }\leq I,$ hence $K_{1}$ with the
corresponding $\mu _{1}$ give the quantum-limited attenuator.
\end{proof}

\section{Entanglement breaking and additivity}

\label{sec:complandadd}

The additivity properties of finite-dimensional entanglement-breaking
channels \cite{shor} were generalized to infinite dimensions in \cite{Sh-2}.
Moreover it was shown in \cite{Sh-2} that the convex closure of the output entropy for such
channel $\Phi $, defined as
\begin{equation}
\hat{S}_{\Phi }\left( \sigma \right) =\inf_{\pi :\bar{\rho}_{\pi }=\sigma
}\int S\left( \Phi \lbrack \rho ]\right) \pi (d\rho ),  \label{convcl}
\end{equation}%
where the infimum is taken over all probability distributions $\pi $ on the
state space with the baricenter $\bar{\rho}_{\pi }\equiv \int \rho \pi
(d\rho )=\sigma ,$ is superadditive, i.e. for an entanglement-breaking
channel $\Phi $ and arbitrary channel $\Psi $
\begin{equation}
\hat{S}_{\Phi \otimes \Psi }\left( \sigma _{12}\right) \geq \hat{S}_{\Phi
}\left( \sigma _{1}\right) +\hat{S}_{\Psi }\left( \sigma _{2}\right)
\label{hiad1}
\end{equation}%
for any state $\sigma _{12}$. This implies additivity of the minimal output
entropy (see e.g. \cite{h}, Proposition 8.15)
\begin{equation}
\min_{\sigma _{12}}S\left( \Phi \otimes \Psi \left[ \sigma _{12}\right]
\right) =\min_{\sigma _{1}}S\left( \Phi \left[ \sigma _{1}\right] \right)
+\min_{\sigma _{2}}S\left( \Psi \left[ \sigma _{2}\right] \right) .
\label{hiad2}
\end{equation}
It turns out
that (\ref{hiad1}) implies
additivity of the constrained $\chi -$capacity of channel $\Phi$, e.g. see \cite{hs}, Sec. 6. Namely,
 let $H$ is positive selfadjoint operator (typically the energy operator), $%
H^{(n)}=H\otimes I\dots \otimes I+\dots +I\otimes \dots \otimes I\otimes H.$
Under mild regularity assumptions (see \cite{Sh-2}) which are fulfilled in
the Gaussian case the constrained $\chi -$capacity is equal to
\begin{equation}
C_{\chi }(\Phi ,H,E)=\sup_{\rho :\mathrm{Tr}\rho H\leq E}\left[ S\left( \Phi
\lbrack \rho ]\right) -\hat{S}_{\Phi }\left( \rho \right) \right] .
\label{cochi1}
\end{equation}%
Then for any entanglement-breaking channel $\Phi $
\begin{equation}
C_{\chi }(\Phi ^{\otimes n},H^{(n)},nE))=nC_{\chi }(\Phi ,H,E),
\label{hiad3}
\end{equation}%
implying

\begin{equation}
C(\Phi ,H,E)=C_{\chi }(\Phi ,H,E),  \label{clca}
\end{equation}%
where $C(\Phi ,H,E)$ is the constrained classical capacity of the channel $%
\Phi $.

Coming to Bosonic system we restrict ourselves to the class of states $%
\mathfrak{S}_{2}$ with finite second moments, satisfying $\mathrm{Tr}\rho
a_{j}^{\dag }a_{j}<\infty ,\,j=1,\dots ,s$ (for rigorous definition of the
moments see e.g. \cite{h}, Sec. 11.1). Notice that this class is invariant
under the action of all Gaussian channels and all the entropy  quantities are finite
for states in this class.
Moreover, one can check by
inspection of proofs that the additivity properties listed above hold for
Gaussian entanglement-breaking channels with optimizations restricted to the
class $\mathfrak{S}_{2},$ provided $\Psi $ is a Gaussian channel and $H$ is
a quadratic Hamiltonian in $a_{j},a_{j}^{\dag },\,\,j=1,\dots ,s.$

\textbf{In what follows we adopt the restriction to the class} $\mathfrak{S}_{2}$
\textbf{without explicit introducing it into notations.}

\begin{proposition}
\label{prop3} The Gaussian gauge-contravariant channel (\ref{CONTRAV}) is
entanglement-breaking.
\end{proposition}

The proof is based on the general criterion of entanglement breaking for
Gaussian channels \cite{h}, Sec. 12.7.2, which in this case amounts to:
there exist a decomposition $\mu =\mu _{1}+\mu _{2}$ such that $\mu _{1}\geq
\frac{1}{2}KK^{\ast },\mu _{2}\geq \frac{1}{2}I.$ This is indeed the case
with $\mu _{1}=\frac{1}{2}KK^{\ast },\mu _{2}=\frac{1}{2}I.$ The
decomposition corresponds to representation of the channel as 1) measurement
of a Gaussian observable and 2) subsequent preparation of coherent states
depending on the outcome of the measurement (see Sec. \ref{sec:new} for
detail).

Proposition \ref{prop3} implies that the additivity properties (\ref{hiad1}%
),(\ref{hiad2}),(\ref{hiad3}) hold for a Gaussian gauge-contravariant
channels $\Phi $ and arbitrary channel $\Psi $.

Denote by $\tilde{\Phi}$ the complementary channel for $\Phi ,$ then it is
known that
\begin{eqnarray}
\min_{\sigma }S\left( \Phi \left[ \sigma \right] \right) &=&\min_{\sigma
}S\left( \tilde{\Phi}\left[ \sigma \right] \right) ,  \label{equality} \\
\hat{S}_{\Phi }\left( \sigma \right) &=&\hat{S}_{\tilde{\Phi}}\left( \sigma
\right) ,  \notag
\end{eqnarray}%
see \cite{h1}, \cite{mkr}, or Sec.6.6.6 of \cite{h}.

\begin{proposition}
\label{prop2} The quantum-limited amplifier with diagonal matrix \ $K$ and
quantum-limited gauge-contravariant channel with diagonal matrix $\sqrt{%
KK^{\ast }-I}\Lambda $ are mutually complementary.
\end{proposition}

In the diagonal case the channels split into tensor products of one-mode
channels. For the proof in the case of one mode see \cite{cgh} or \cite{h},
Sec. 12.6.1.

Via Proposition \ref{prop2} and (\ref{equality}), the additivity properties (%
\ref{hiad1}), (\ref{hiad2}) then also hold for the diagonal quantum-limited
amplifier $\Phi $ and arbitrary channel $\Psi $.

\section{Reductions of the Gaussian optimizer conjecture}

\label{sec:reduction}

For the multi-mode quantum-limited attenuator $\Phi _{1}$ the vacuum $\rho
_{vac}$ is invariant state, hence the minimal output entropy, equal to zero,
is attained on the vacuum (as well as on any coherent state). Using
Proposition \ref{prop1} and the argument from \cite{gp,QCMC} allows to reduce
solution of the minimal entropy conjecture for arbitrary channel satisfying
the conditions of Proposition to the case of quantum-limited amplifier.
Since quantum-limited amplifier $\Phi _{2}$ is additive, this argument also
implies the additivity property of the minimal output entropy. More
precisely,

\begin{proposition}
\label{prop4} Under the hypothesis:

\textbf{(A) The minimal output entropy of any one-mode quantum-limited
amplifier is attained on the vacuum state},

the following relation holds
\begin{equation}
\min_{\rho ^{(n)}}S\left( \Phi ^{\otimes n}\left[ \rho ^{(n)}\right] \right)
=nS\left( \Phi \left[ \rho _{vac}\right] \right),  \label{minent}
\end{equation}%
for any Gaussian gauge-covariant or contravariant channel $\Phi $.

Moreover, for any collection of channels $\Phi _{1},\dots ,\Phi _{N}$
satisfying the condition of the Proposition \ref{prop1}
\begin{equation}
\min_{\rho ^{(n)}}S\left( (\otimes _{j=1}^{N}\Phi _{j})\left[ \rho ^{(n)} %
\right] \right) =\sum_{j=1}^{N}\min_{\rho _{j}}S\left( \Phi _{j}\left[ \rho
_{j}\right] \right) =\sum_{j=1}^{N}S\left( \Phi _{j}\left[ \rho _{_{j,vac}} %
\right] \right) .  \label{hiad4}
\end{equation}
\end{proposition}

\begin{proof}
First, let us notice that by the hypothesis (A) and the additivity property (%
\ref{hiad2}) we have for any multi-mode diagonal quantum-limited amplifier $%
\Phi =\otimes _{j=1}^{s}\Phi _{j},$ where $\Phi _{j}$ are one-mode
quantum-limited amplifiers,
\begin{equation}
\min_{\rho }S\left( \Phi \left[ \rho \right] \right)
=\sum_{j=1}^{s}\min_{\rho _{j}}S\left( \Phi _{j}\left[ \rho _{j}\right]
\right) =\sum_{j=1}^{s}S\left( \Phi _{j}\left[ \rho _{j,vac}\right] \right)
=S\left( \Phi \left[ \rho _{vac}\right] \right) .  \label{A_s}
\end{equation}%
It follows that the analog of property (A) holds for any multi-mode
diagonalizable quantum-limited amplifier because for such channel%
\begin{equation*}
\Phi \left[ \rho \right] =U^{\ast }(\otimes _{j=1}^{N}\Phi _{j})\left[ U\rho
U^{\ast }\right] U,
\end{equation*}%
where $U$ is the unitary operator in $\mathcal{H}^{\otimes N}$ implementing
the unitary transformation $R$ in $Z$ which diagonalizes $K.$ Notice also that $%
U\rho _{vac}U^{\ast }=\rho _{vac}$ because $U^{\ast }\mathbf{a}U=R\mathbf{a}%
. $

Therefore by complementarity (Proposition \ref{prop2} and (\ref{equality}))
similar property holds for diagonalizable quantum-limited
gauge-contravariant channel.

Let us prove (\ref{hiad4}) since (\ref{minent}) is a particular case. For
every channel $\Phi _{j}$ we have decomposition $\Phi _{j}=\Phi _{2,j}\circ
\Phi _{1,j},$ where $\Phi _{1,j}$ is quantum-limited attenuator and $\Phi
_{2,j}$ has the property (\ref{A_s}). Then
\begin{equation*}
\sum_{j=1}^{N}S\left( \Phi _{j}\left[ \rho _{_{j,vac}}\right] \right)
=\sum_{j=1}^{N}S\left( \Phi _{2,j}\left[ \rho _{_{j,vac}}\right] \right) ,
\end{equation*}%
where the invariance of vacuum state under $\Phi _{1,j}$ was used. By (\ref%
{A_s}) this is equal to
\begin{equation*}
\sum_{j=1}^{N}\min_{\sigma _{j}}S\left( \Phi _{2,j}\left[ \sigma _{j}\right]
\right) .
\end{equation*}%
By the additivity (\ref{hiad2}) of the minimal output entropy for Gaussian
gauge-contravariant channels and complementary quantum-limited amplifiers,
the last sum is equal to
\begin{eqnarray*}
\min_{\sigma ^{(n)}}S\left( (\otimes _{j=1}^{N}\Phi _{2,j})\left[ \sigma
^{(n)}\right] \right) &\leq &\min_{\rho ^{(n)}}S\left( (\otimes
_{j=1}^{N}\Phi _{j})\left[ \rho ^{(n)}\right] \right) \\
&\leq &\sum_{j=1}^{N}S\left( \Phi _{j}\left[ \rho _{_{j,vac}}\right] \right)
,
\end{eqnarray*}%
hence (\ref{hiad4}) follows.
\end{proof}

We now turn to the classical capacity. In \cite{h} Proposition 12.39 it is
shown that for a Gaussian gauge-covariant channel $\Phi $, the validity of
the minimal output entropy conjecture implies positive solution of the
optimal Gaussian ensemble conjecture for the $\chi -$capacity under the
energy constraint \textrm{Tr}$\rho H\leq E$ with gauge-invariant oscillator
Hamiltonian $H=\sum_{j,k=1}^{s}a_{j}^{\dag }\epsilon _{jk}a_{k},$ where $ \epsilon =
\left[ \epsilon _{jk}\right] $ is a positive definite matrix. The resulting
expression is
\begin{eqnarray}
C_{\chi }(\Phi ,H,E)&=&\max_{\rho :\mathrm{Tr}\rho H\leq E}S\left( \Phi %
\left[ \rho \right] \right) -\min_{\rho }S\left( \Phi \left[ \rho \right]
\right)  \notag \\
&=&\max_{\rho :\mathrm{Tr}\rho H\leq E}S\left( \Phi \left[ \rho \right]
\right) -S\left( \Phi \left[ \rho _{vac}\right] \right) .  \label{cchi}
\end{eqnarray}
Then, under the hypothesis (A), one shows the additivity of $C_{\chi }(E)$
similarly to the case of one-mode pure loss channel in \cite{gio}:

\begin{eqnarray*}
&&n\left[ \max_{\rho :\mathrm{Tr}\rho H\leq E}S\left( \Phi \left[ \rho %
\right] \right) -S\left( \Phi \left[ \rho _{vac}\right] \right) \right] \leq
nC_{\chi }(\Phi ,H,E) \leq C_{\chi }(\Phi ^{\otimes n},H^{(n)},nE) \\
&&\qquad \qquad \leq \max_{\rho ^{(n)}:\mathrm{Tr}\rho ^{(n)}H^{(n)}\leq
nE}S\left( \Phi ^{\otimes n}\left[ \rho ^{(n)}\right] \right) -\min_{\rho
^{(n)}}S\left( \Phi ^{\otimes n}\left[ \rho ^{(n)}\right] \right) \\
&&\qquad \qquad \leq n\left[ \max_{\rho :\mathrm{Tr}\rho H\leq E}S\left(
\Phi \left[ \rho \right] \right) -S\left( \Phi \left[ \rho _{vac}\right]
\right) \right] ,
\end{eqnarray*}
where in the last inequality we used Lemma 11.20 \cite{h} and (\ref{minent}).

Thus $C_{\chi }(\Phi ^{\otimes n},H^{(n)},nE)=nC_{\chi }(\Phi ,H,E)$ and
hence the constrained classical capacity $C(\Phi ,H,E)=\lim_{n\rightarrow
\infty }\frac{1}{n}C_{\chi }(\Phi ^{\otimes n},H^{(n)},nE)$ of the Gaussian
gauge-covariant channel is given by the same expression (\ref{cchi}).

Now let us use the familiar formula for the entropy of Gaussian state (\ref
{gausstate}) $S\left( \rho \right) =\mathrm{tr\,}g(\alpha -I/2),$ where $g(x)=(x+1)\log (x+1)-x\log x,$ and $\mathrm{tr}$ denotes trace of operators
in $\mathbf{Z}.$ Applying the transformation rule $\alpha \rightarrow
K\alpha K^{\ast }+\mu $ of the covariance matrix $\alpha $ under the action
of the channel (\ref{defprima}), see e.g. \cite{h} Ch. 12, to the vacuum
state with $\alpha =I/2,$ we obtain explicitly the minimal output entropy
\begin{equation}\label{entro}
\min_{\rho }S(\Phi ^{\otimes n}[\rho ])=n\,\mathrm{tr\,}g(\mu +\left(
KK^{\ast }-I\right) /2).
\end{equation}
For the classical capacity the relation (\ref{cchi}) gives%
\begin{eqnarray*}
C(\Phi ;H,E) &=&C_{\chi }(\Phi ;H,E) \label{capa} \\
&=&\max_{\nu \in \Sigma _{E}}\,\mathrm{tr\,}g(K\nu K^{\ast }+\mu +\left(
KK^{\ast }-I\right) /2)-g(\mu +\left( KK^{\ast }-I\right) /2),\nonumber
\end{eqnarray*}
where $\Sigma _{E}=\left\{ \nu :\nu \geq 0,\,\mathrm{tr}\nu \epsilon \leq
E\right\} $ is the set of covariance matrices $\nu $ of the Gaussian
ensemble with the energy constraint. This reduces to a finite-dimensional
optimization problem which is a quantum analog of
``water-filling'' problem in classical information theory, see e.g. \cite{COVER,h2,WATER,BROAD,JOA,OLEG}. It
can be solved explicitly only in some special cases, e.g. when $K,\mu ,\epsilon $ commute, and it is a subject of separate study.

Similar argument applies to Gaussian gauge-contravariant channel (\ref{CONTRAV}).

\section{Some useful properties of contravariant channels}

\label{sec:new}

From the previous sections it follows that minimal output entropy problem
for an arbitrary gauge-covariant channel $\Phi $ can be solved by proving
that the minimal output entropy conjecture for a single-mode quantum-limited
amplifier. Alternatively, due to the identity~(\ref{equality}) and to
Proposition~\ref{prop2}, this is equivalent to show that single-mode
quantum-limited contravariant map $\tilde{\Phi}$ has the vacuum as minimizer
of output entropy. We start hence by providing a characterization of the
output states of these maps. Even though for the proof of the conjecture we
need only the single-mode case, for the sake of generality we will present
them in the multi-mode scenario. Denote by $|z\rangle =D(z)|0\rangle ,\,z\in
\mathbf{Z}$, the multimode coherent states. From now on we skip the
subscripts $A,B$ occasionally, since it should be clear from the context.

\begin{proposition}
\label{prop5} Given a quantum-limited Gaussian gauge-contravariant channel $%
\Phi $ described by the matrix $K$ via Eqs.~(\ref{CONTRAV}) and (\ref%
{mindef5}), its action on an input state $\rho $ can be expressed as the
following measure-reprepare mapping
\begin{equation}
\rho \mapsto \Phi \lbrack \rho ]=\int \frac{d^{2s}z}{\pi ^{s}}|-K\bar{z}%
\rangle \langle -K\bar{z}|\;{\langle }z|\rho |z\rangle .  \label{EQ1}
\end{equation}
\end{proposition}

\begin{proof}
The relation (\ref{EQ1}) follows from the general measure-reprepare
representation of Gaussian entanglement-breaking channels, \cite{h}, Sec.
12.7.2. For a direct verification of~(\ref{EQ1}) it is sufficient to show
that Husimi functions (diagonal values in the coherent-state representation)
coincide for operator (\ref{CONTRAV}) and the dual of (\ref{EQ1}). This
amounts to the identity (where we redenoted some variables)
\begin{equation*}
{\langle }u|D(\Lambda K^{\ast }w)|u\rangle\; e^{-\frac{1}{2}w^{\ast }\left(
I+KK^{\ast }\right) w }=\int \frac{d^{2s}z}{\pi ^{s}}\langle -K\bar{z}%
|D(w)|-K\bar{z}\rangle \;\left\vert {\langle }u|z\rangle \right\vert ^{2},
\end{equation*}%
which is verified by using the formulas ${\langle }u|D(w)|u\rangle =\exp
\left( 2i\Im\bar{u}w-|w|^{2}/2\right) $ and $\left\vert {\langle }u|z\rangle
\right\vert ^{2}=\exp \left( -|w-z|^{2}\right) .$
\end{proof}

For any channel (\ref{EQ1}) we introduce the following \textit{skewed}
counterparts defined as
\begin{eqnarray}
\rho \mapsto \Psi _{+}[\rho ] &=&\int \frac{d^{2s}z}{\pi ^{s}}|\bar{K}{z}\rangle
\langle \bar{K}{z}|\;{\langle }z|\rho |z\rangle ,  \label{PSI+} \\
\rho \mapsto \Psi _{-}[\rho ] &=&\int \frac{d^{2s}z}{\pi ^{s}}|-\bar{K}{z}\rangle
\langle -\bar{K}{z}|\;{\langle }z|\rho |z\rangle .
\end{eqnarray}%
These are again measure-reprepare channels where, differently from $\Phi $
of Eq.~(\ref{EQ1}), after a measurement outcome $z,$ the output system is
initialized into the coherent states $|\pm \bar{K} z\rangle $. Similarly to
Proposition~\ref{prop5} one can verify that in the Heisenberg representation
the channels $\Psi _{\pm }$ are described by the gauge-covariant mappings
\begin{equation}
\Psi _{\pm }^{\ast }[D(z)]=D(\pm \bar{K}z)\exp \left( -z^{\ast }(\bar{K}\bar{K}^{\ast
}+I)z/2\right) .  \label{PSI++}
\end{equation}%
However these channels are no longer quantum-limited.

\begin{proposition}
\label{prop6} The output states $\Psi _{+}[\rho ]$, $\Psi _{-}[\rho ]$, and $%
\Phi \lbrack \rho ]$ have the same eigenvalues and hence the same entropy.
\end{proposition}

\begin{proof}
Noticing that $T[|{z}\rangle \langle {z}|]=|{\bar{z}}\rangle \langle {\bar{z}%
}|$, it follows that $\Psi _{-}[\rho ]=T[\Phi \lbrack \rho ]]$. Therefore $%
\Psi _{-}[\rho ]$ and $\Phi \lbrack \rho ]$ must have the same spectrum. To
prove that also $\Psi _{+}[\rho ]$ shares the same property, notice that we
can transform such state into $\Psi _{-}[\rho ]$ by a phase transformation $%
e^{-i\pi N}$. Indeed, $e^{-i\pi N}|z\rangle =|e^{-i\pi }z\rangle =|-z\rangle
$, hence
\begin{equation*}
e^{-i\pi N}\Psi _{+}[\rho ]e^{i\pi N}=\Psi _{-}[\rho ].
\end{equation*}
\end{proof}

We conclude that given a quantum-limited contravariant channel $\Phi $ and
an input state $\rho $, there exists a unitary transformation $U$ (possibly
dependent upon $\rho $) such that
\begin{equation*}
\Phi \lbrack \rho ]=U\Psi _{+}[\rho ]U^{\ast }.
\end{equation*}%
We have now all the elements we need to prove the minimum output entropy
conjecture. For this purpose we have to take a step back and re-introduce
the complementary counterpart of the contravariant channel.

\section{Proof of the Gaussian optimizer conjecture}

\label{sec:6} In this Section we prove the hypothesis (A) in Proposition~\ref%
{prop4}.

Here $\Phi $ and $\tilde{\Phi}$ are the single-mode quantum-limited
covariant amplifier, resp. contravariant channel defined by the parameter $%
K\geq 1$ via the relations
\begin{eqnarray}
\Phi ^{\ast }[D(z)] &=&D(Kz)\exp \left( -z^{\ast }(K^{2}-1)z/2\right) ,
\label{eq50} \\
\tilde{\Phi}^{\ast }[D(z)] &=&D(G\bar{z})\exp \left( -z^{\ast
}(G^{2}+1)z/2\right) ,
\end{eqnarray}%
where $G=\sqrt{K^{2}-1}.$ In what follows we can assume that $K>1$ (the case
$K=1$ is trivial, corresponding to the identity channel). Accordingly $G $
is strictly positive and we can apply to the contravariant channel $\tilde{%
\Phi}$ all the results we have derived in the previous Sections. From
Proposition~\ref{prop2} we know that $\Phi $ and $\tilde{\Phi}$ are mutually
complementary hence  the density operators $\Phi (|\psi \rangle \langle \psi |)$ and $\tilde{\Phi}%
(|\psi \rangle \langle \psi |)$ have the same nonzero spectrum, hence there exists a
partial isometry $\tilde{V}$ (possibly dependent upon the input state $|\psi \rangle$), mapping the support of one operator onto the support of another, such that
\begin{equation*}\label{ff111}
\Phi \left[ |\psi \rangle \langle \psi |\right] =\tilde{V}\tilde{\Phi}\left[
|\psi \rangle \langle \psi |\right] \tilde{V}^{\ast }.
\end{equation*}
In such a case we will call the partial isometry \textit{connecting} the relevant density operators. Remind that the connected operators have equal entropies.
Furthermore from Proposition~\ref{prop6} it also follows that an analogous
relation connects $\tilde{\Phi}$ and $\Psi _{+}$. Therefore for any $|\psi
\rangle $ there exists a connecting partial isometry $V$ (possibly dependent upon
$ |\psi \rangle $) such that
\begin{equation}
\Phi \left[ |\psi \rangle \langle \psi |\right] =V{\Psi _{+}}\left[ |\psi
\rangle \langle \psi |\right] V^{\ast },  \label{eq51}
\end{equation}
where $\Psi _{+}$ is the channel~(\ref{PSI++}) associated with the
quantum-limited contravariant channel $\tilde{\Phi}$, i.e.
\begin{equation}
\Psi _{+}^{\ast }[D(z)]=D(Gz)\exp \left( -z^{\ast }(G^{2}+1)z/2\right) =D(%
\sqrt{K^{2}-1}\;z)\exp \left( -z^{\ast }K^{2}z/2\right) .
\end{equation}
As already noticed the channel $\Psi _{+}$ is in general not
quantum-limited. Nevertheless, following Proposition \ref{prop1}, we can
express it as a concatenation of a quantum-limited attenuator $\Phi _{1}$
followed by a quantum-limited covariant amplifier $\Phi _{2}$, i.e.$\Psi
_{+}=\Phi _{2}\circ \Phi _{1}.$ The parameters $K_{2}$ and $K_{1}$ which
define these maps can be computed as
\begin{eqnarray}
K_{2} &=&\sqrt{K^{2}/2+\frac{G^{2}+1}{2}}=K,  \label{K2} \\
K_{1} &=&G/K=\frac{\sqrt{K^{2}-1}}{K}. \label{KK1KK}
\end{eqnarray}
From Eq.~(\ref{K2}) it follows that $\Phi _{2}$ is nothing but the channel $%
\Phi $ we started from, hence $\Psi _{+}=\Phi \circ \Phi _{1}.$ Therefore
substituting this into Eq.~(\ref{eq51}) we get
\begin{equation}
\Phi \left[ |\psi \rangle \langle \psi |\right] =V\left( {\Phi \circ \Phi
_{1}}\right) \left[ |\psi \rangle \langle \psi |\right] V^{\ast },
\label{IMPO1}
\end{equation}
which applies for all pure inputs $|\psi \rangle $ (we remind that the
connecting partial isometry $V$ can in principle depend upon $|\psi \rangle $).

It is worth observing that from Eq.~(\ref{IMPO1}) it follows that the
minimal output entropy of the quantum-minimal amplifier $\Phi $ coincides
with the minimal output entropy of $\ \Phi \circ \Phi _{1}$ with $\Phi _{1}$
being the attenuator~(\ref{KK1KK}), a fact which is fully consistent with
the conjecture since $\Phi _{1}$ admits the vacuum as fixed point. We can
however say more.

\begin{proposition}
\label{propo10} Let $\Phi $ the quantum-limited covariant amplifier of Eq.~(%
\ref{eq50}) and $\Phi _{1}$ the attenuator channel associated to it through
Eq.~(\ref{KK1KK}). Then given a pure input state $\rho = |\psi \rangle \langle \psi |$ and integer $%
n $, there exists an ensemble $\mathcal{E}=\{p_{i};|\psi _{i}\rangle \}$ and
connecting partial isometries $U_{i}$ satisfying the relations
\begin{eqnarray}
\Phi _{1}^{n}(\rho ) &=&\sum_{i}p_{i}|\psi _{i}\rangle \langle \psi _{i}|,
\label{IM1} \\
\Phi \lbrack \rho ] &=&\sum_{i}p_{i}U_{i}\Phi \lbrack |\psi _{i}\rangle
\langle \psi _{i}|]U_{i}^{\ast }.  \label{IM2}
\end{eqnarray}
\end{proposition}

\begin{proof}
We prove this by induction. Let
\begin{equation*}
\Phi _{1}(|\psi \rangle \langle \psi |)=\sum_{j}p_{j}|\psi _{j}\rangle
\langle \psi _{j}|\;
\end{equation*}%
be the spectral decomposition of the state $\Phi _{1}(|\psi \rangle \langle
\psi |)$, where $p_{j}$ are strictly positive (zero eigenvalues can be
omitted as they do not contribute to the sum). Inserting this into~(\ref%
{IMPO1}) we get
\begin{equation*}
\Phi (|\psi \rangle \langle \psi |)=\sum_{j}p_{j}V\Phi (|\psi _{j}\rangle
\langle \psi _{j}|)V^{\ast },
\end{equation*}%
proving the statement for $n=1.$

Assume now that the statement is valid for some $n.$ Then from (\ref{IM1})%
\begin{eqnarray*}
\Phi _{1}^{n+1}(|\psi \rangle \langle \psi |) &=&\sum_{i}p_{i}\Phi _{1}\left[
|\psi _{i}\rangle \langle \psi _{i}|\right] \\
&=&\;\sum_{i,j}p_{i}p_{j|i}|\psi _{j|i}\rangle \langle \psi _{j|i}|,
\end{eqnarray*}
where $\Phi _{1}\left[ |\psi _{i}\rangle \langle \psi _{i}|\right]
=\sum_{,j}p_{j|i}|\psi _{j|i}\rangle \langle \psi _{j|i}|$ is the spectral
decomposition. By using (\ref{IM2}), (\ref{IMPO1}) we obtain%
\begin{eqnarray*}
\Phi (|\psi \rangle \langle \psi |) &=&\sum_{i}p_{i}U_{i}(\Phi \lbrack |\psi
_{i}\rangle \langle \psi _{i}|])U_{i}^{\ast } \\
&=&\sum_{i}p_{i}U_{i}V_{i}\left( \Phi \circ \Phi _{1}\right) ([|\psi
_{i}\rangle \langle \psi _{i}|])\left[ U_{i}V_{i}\right] ^{\ast }\; \\
&=&\sum_{i,j}p_{j}p_{j|i}\left[ U_{i}V_{i}\right] \Phi (|\psi _{j|i}\rangle
\langle \psi _{j|i}|)\left[ U_{i}V_{i}\right] ^{\ast },
\end{eqnarray*}
proving the statement for $n+1.$
\end{proof}

\bigskip We can then use the concavity of the von Neumann entropy to write
the inequality
\begin{equation}
S(\Phi \lbrack \rho ])\geq \sum_{i}p_{i}S(\Phi (|\psi _{i}\rangle \langle
\psi _{i}|)),  \label{IMPO1n1}
\end{equation}%
where the ensemble satisfies (\ref{IM1}). Notice also for $n\rightarrow
\infty $ the channel $\Phi _{1}^{n}$ brings all the states into the fixed
point, i.e. the vacuum state,
\begin{equation*}
\lim_{n\rightarrow \infty }\Phi _{1}^{n}(|\psi \rangle \langle \psi
|)=|0\rangle \langle 0|,
\end{equation*}%
the convergence being in trace-norm. Accordingly as $n\rightarrow \infty $
the only state surviving in the decomposition (\ref{IM1}) is the vacuum
state. It seems then reasonable to conclude that in the limit $n\rightarrow
\infty $ the right-hand side of Eq.~(\ref{IMPO1n1}) should reduce to $S(\Phi
(|\psi \rangle \langle \psi |))\geq S(\Phi (|0\rangle \langle 0|))$ hence
proving the thesis. To make this precise we use the monotonicity of the
relative entropy~\cite{h}, i.e.
\begin{equation*}
S(|\psi _{i}\rangle \langle \psi _{i}|||\sigma )\geq S(\Phi (|\psi
_{i}\rangle \langle \psi _{i}|)||\Phi (\sigma ))
\end{equation*}%
where $|\psi _{i}\rangle $ is one of the vectors of the ensemble for $\Phi
_{1}^{n}$, and $\sigma $ a state to be defined later. By reorganizing
various terms this gives
\begin{equation*}
S(\Phi (|\psi _{i}\rangle \langle \psi _{i}|))\geq -\mbox{Tr}\Phi (|\psi
_{i}\rangle \langle \psi _{i}|)\log \Phi (\sigma )+\mbox{Tr}|\psi
_{i}\rangle \langle \psi _{i}|\log \sigma ,
\end{equation*}%
which, substituted into (\ref{IMPO1n1}), yields
\begin{eqnarray}
S(\Phi (|\psi \rangle \langle \psi |)) &\geq &-\mbox{Tr}\Phi
(\sum_{i}p_{i}|\psi _{i}\rangle \langle \psi _{i}|)\log \Phi (\sigma )+%
\mbox{Tr}\sum_{i}p_{i}|\psi _{i}\rangle \langle \psi _{i}|\log \sigma
\notag \\
&=&-\mbox{Tr}\left( \Phi \circ \Phi _{1}^{n}\right) (|\psi \rangle \langle
\psi |)\log \Phi (\sigma )]+\mbox{Tr}\Phi _{1}^{n}(|\psi \rangle \langle
\psi |)\log \sigma   \notag \\
&=&S(\left( \Phi \circ \Phi _{1}^{n}\right) (|\psi \rangle \langle \psi
|))+S(\left( \Phi \circ \Phi _{1}^{n}\right) (|\psi \rangle \langle \psi
|)||\Phi (\sigma ))  \notag \\
&&+\mbox{Tr}\Phi _{1}^{n}(|\psi \rangle \langle \psi |)\log \sigma   \notag
\\
&\geq &S(\left( \Phi \circ \Phi _{1}^{n}\right) (|\psi \rangle \langle \psi
|))+\mbox{Tr}\Phi _{1}^{n}(|\psi \rangle \langle \psi |)\log \sigma .
\label{FFF}
\end{eqnarray}%

Assume next that $\sigma $ is a Gibbs state, i.e.
\begin{equation*}
\sigma =(1-\gamma )\sum_{k=0}^{\infty }\gamma ^{k}|k\rangle \langle k|,
\end{equation*}%
with $\gamma >0$. With this choice the second term of Eq.~(\ref{FFF}) can be
well defined for all input states $|\psi \rangle \langle \psi |$ having
finite second moments. Indeed by repeated use of the relation%
\begin{equation*}
\mbox{Tr}\Phi _{1}(\rho )a^{\dag }a=\left[ \frac{K^{2}-1}{K^2}\right] \mbox{Tr}%
\rho a^{\dag }a
\end{equation*}%
valid for states from $\mathfrak{S}_{2}$, we have
\begin{equation*}
\mbox{Tr}[\Phi _{1}^{n}(|\psi \rangle \langle \psi |)\log \sigma ]=\log
(1-\gamma )+\log \gamma \;\left[ \frac{K^{2}-1}{K^2}\right] ^{n}\langle \psi
|a^{\dag }a|\psi \rangle \;.
\end{equation*}%
Substituting this into the right-hand-side of Eq.~(\ref{FFF}) gives
\begin{equation*}
S(\Phi (|\psi \rangle \langle \psi |))\geq S(\left( \Phi \circ \Phi
_{1}^{n}\right) (|\psi \rangle \langle \psi |))+\log (1-\gamma )+\log \gamma
\;\left[ \frac{K^{2}-1}{K^2}\right] ^{n}\langle \psi |a^{\dag }a|\psi \rangle
\;.
\end{equation*}%
Taking the limit $n\rightarrow \infty $ and using lower semicontinuity of
the quantum entropy in the first term we obtain
\begin{equation*}
S(\Phi (|\psi \rangle \langle \psi |))\geq S(\Phi (|0\rangle \langle
0|))+\log (1-\gamma ).
\end{equation*}%
Taking the limit $\gamma \rightarrow 0$ we finally have%
\begin{equation*}
S(\Phi (|\psi \rangle \langle \psi |))\geq S(\Phi (|0\rangle \langle 0|))
\end{equation*}%
for all input states $|\psi \rangle \langle \psi |$ with finite second
moments.

\section{Implications and perspectives}

\label{sec:conc}

In this work we have proven that the minimal entropy at the output of a
(possibly multimode) BGC covariant (or contravariant) channel $\Phi $ is
achieved by the vacuum input state, restricting our analysis \ to the class
of state with finite second moments. As detailed in Sec.~\ref{sec:reduction}
this implies both the additivity of the minimal output entropy functional
and of the classical capacity (under energy constraint) whose value is
achieved via Gaussian encodings. To be specific, let us apply the formulas (\ref{entro}), (\ref{capa})
to single-mode ($s=1$) quantum channels~\cite{cgh} to obtain the quantum
counterparts of the Shannon formula (\ref{shan}). In this case one
identifies three classes of covariant maps:

\begin{itemize}
\item The \textit{thermal noise channels} describing a passive exchange
(beam splitter) interaction with an external Gibbs thermal state. Following
the notation of~\cite{conj1} they are characterized by two real parameters $%
\eta \in \lbrack 0,1]$ and $N\in \lbrack 0,\infty \lbrack $, associated
respectively to the intensity of exchange coupling and to the temperature of
the system environment, and which enter into Eq.~(\ref{defprima}) as $K=%
\sqrt{\eta }$ and $\mu =(1-\eta )(N+1/2)$. For these channels our result
shows that
\begin{eqnarray*}
\min_{\rho }S(\Phi ^{\otimes n}[\rho ]) &=&n\;g((1-\eta )N),\qquad \qquad \\
C(\Phi ;H,E) &=&C_{\chi }(\Phi ;H,E) \\
&=&g(\eta E+(1-\eta )N)-g((1-\eta )N).
\end{eqnarray*}

\item The \textit{additive classical noise channels} which randomly displace
the input states in phase space. They can be fully characterized via a
single parameter $N\in \lbrack 0,\infty \lbrack $ which represents the
variance of the Gaussian probability distribution governing the displacement
transformation, and which enters into Eq.~(\ref{defprima}) via the
identities $K=1$, $\mu =N$. In this case we have
\begin{eqnarray*}
\min_{\rho }S(\Phi ^{\otimes n}[\rho ]) &=&n\;g(N),\qquad \qquad \\
C(\Phi ;H,E) &=&C_{\chi }(\Phi ;H,E)=g(E+N)-g(N).
\end{eqnarray*}

\item The \textit{noisy amplifier channels} characterized by two real
parameters $\kappa \in \lbrack 1,\infty \lbrack $ and $N\in \lbrack 0,\infty
\lbrack $ entering Eq.~(\ref{defprima}) via the identities $K=\sqrt{\kappa }$
and $\mu =(\kappa -1)(N+1/2)$. In this case we have
\begin{eqnarray*}
\min_{\rho }S(\Phi ^{\otimes n}[\rho ]) &=&n\;g((\kappa -1)(N+1)),\qquad
\qquad \\
C(\Phi ;H,E) &=&C_{\chi }(\Phi ;H,E) \\
&=&g(\kappa E+(\kappa -1)(N+1))-g((\kappa -1)(N+1)).
\end{eqnarray*}
\end{itemize}

In the single mode scenario, all BGCs with nondegenerate $K$ are unitarily
equivalent to covariant (or contravariant) channels via metaplectic unitary
transformations, while for degenerate channels the proof is even simpler
\cite{conj2}. Thus one can easily verify that the optimal states which
minimize the output entropy are squeezed vacuum and the corresponding
coherent states, a result which in certain regimes allows also to compute
the constrained classical capacity. Most importantly,
with Proposition~\ref{propo10} it is possible to prove~\cite{ANDREA} that,
for arbitrary covariant (contravariant) channels the vacuum input produces the output which majorizes
all other output states (a result which implies the minimal output
conjecture).

In the multimode case, the argument of this paper concerning the minimal
output entropy applies to BGCs gauge-covariant with respect to any
\textquotedblleft squeezed\textquotedblright\ complex structure in the
underlying real symplectic space, since it can be always reduced to the
standard one considered in this paper. However the problem with the
classical capacity arises already in the case (perhaps artificial from the
point of view of applications) where the complex structures associated with
the gauge-covariant BGC and the energy operator do not agree. \newline

\textit{Acknowledgements}: The authors are grateful to R. F. Werner, J.
Oppenheim, A. Winter, A. Mari, L. Ambrosio, and M. E. Shirokov for comments and discussions. They also acknowledge
support and catalysing role of the Isaac Newton Institute for Mathematical
Sciences, Cambridge, UK: important part of this work was conducted when
attending the Newton Institute programme \textit{Mathematical Challenges in
Quantum Information}. AH acknowledges the Rothschild Distinguished Visiting Fellowship which enabled him to participate
in the programme and partial support from RAS Fundamental
Research Programs, Russian Quantum Center and RFBR grant No 12-01-00319.
RG-P acknowledge financial support from the F.R.S.-FNRS and
 from the Alexander von Humboldt Foundation.


\begin{thebibliography}{99}
\bibitem{COVER} T. M. Cover and J. A. Thomas, 
Elements of Information Theory, (New York, Wiley, 1968).

\bibitem{Shannon} C. E. Shannon, 
A Mathematical Theory of Communication,
Bell Syst. Tech. J. \textbf{27} (1948), 379.

\bibitem{BENSHOR} C. H. Bennett and P. W. Shor, 
Quantum Information Theory,
IEEE Trans. Inf. Theory \textbf{44}
(1998), 2724.

\bibitem{HG} A. S. Holevo and V. Giovannetti, 
Quantum channels and their entropic characteristics, 
Rep. Prog. Phys. \textbf{75}
(2012), 046001.

\bibitem{HOLEVO98} A. S. Holevo, 
The Capacity of the Quantum Channel with General Signal States, 
IEEE Trans. Inf. Theory \textbf{44} (1998), 269.

\bibitem{SCHWEST} B. Schumacher and W. D. Westmoreland, 
Sending classical information via noisy quantum channels, 
Phys. Rev. A \textbf{56} (1997), 131.

\bibitem{HOWE} A. S. Holevo and R. F. Werner, 
Evaluating capacities of bosonic Gaussian channels, 
Phys. Rev. A \textbf{63} (2001), 032312.

\bibitem{CAVES} C. M. Caves and P. B. Drummond, 
Quantum limits on bosonic communication rates, 
Rev. Mod. Phys. \textbf{66} (1994), 481.

\bibitem{h2} A. S. Holevo, 
Quantum coding theorems, 
Russian Math. Surveys \textbf{53} (1998), 1295-1331.

\bibitem{HOSHI} A. S. Holevo and M. E. Shirokov, 
Continuous ensembles and the $\chi$-capacity of infinite-dimensional channels, 
Probab. Theory and Appl. \textbf{50} (2005), 86-98.

\bibitem{h} A. S. Holevo, 
Quantum systems, channels, information. A mathematical introduction, 
(De Gruyter, Berlin--Boston, 2012).

\bibitem{conj1} V. Giovannetti, S. Guha, S. Lloyd, L. Maccone, and J. H.
Shapiro, Minimum output entropy of bosonic channels: A conjecture, Phys. Rev. A \textbf{70}
(2004), 032315.

\bibitem{conj2} V. Giovannetti, A. S. Holevo, S. Lloyd, and L. Maccone,
Generalized minimal output entropy conjecture for one-mode Gaussian
channels: definitions and some exact results, 
J. Phys. A \textbf{43} (2010), 415305.

\bibitem{conj3} V. Giovannetti, S. Guha, S. Lloyd, L. Maccone, and J. H.
Shapiro, Minimum bosonic channel output entropies, 
AIP Conf. Proc. \textbf{734} (2004), 21.

\bibitem{gp} R. Garc{\'i}a-Patr\'on, C. Navarrete-Benlloch, S. Lloyd, J. H.
Shapiro, and N. J. Cerf, Majorization theory approach to the Gaussian
channel minimum entropy conjecture, 
Phys. Rev. Lett. \textbf{108} (2012), 110505; arXiv:1111.1986 [quant-ph].

\bibitem{konig1} R. K\"{o}nig and G. Smith, Classical capacity of quantum
thermal noise channels to within 1.45 Bits, 
Phys. Rev. Lett. \textbf{110} (2013),  040501.

\bibitem{konig2} R. K\"{o}nig and G. Smith, Limits on classical
communication from quantum entropy power inequalities, 
Nature Photon. \textbf{7} (2013), 142.

\bibitem{natphot} V. Giovannetti, S. Lloyd, L. Maccone, and J. H. Shapiro,
Nature Photon. \textbf{7} (2013),  834.

\bibitem{gio} V. Giovannetti, S. Guha, S. Lloyd, L. Maccone, J. H. Shapiro,
and H. P. Yuen, Classical capacity of the lossy bosonic channel: the exact
solution, Phys. Rev. Lett. \textbf{92} (2004), 027902;
arXiv:quant-ph/0308012.

\bibitem{ANDREA} A. Mari, V. Giovannetti, and A. S. Holevo, Quantum state majorization at the output of bosonic Gaussian channels, arXiv:1312.3545.

\bibitem{walls} D. F. Walls and G. J. Milburn, {Quantum Optics} 
(Springer, Berlin, 1994).

\bibitem{htw} T. Heinosaari, A. S. Holevo, and M. M. Wolf, The semigroup
structure of Gaussian channels, Quantum Inf. Comp. \textbf{10} (2010),
0619-0635; arXiv:0909.0408.

\bibitem{cgh} F. Caruso, V. Giovannetti, and A. S. Holevo, 
One-mode Bosonic Gaussian channels: 
a full weak-degradability classification, 
New J. Phys.  \textbf{8} (2006), 310; 
arXiv:quant-ph/0609013.

\bibitem{shor} P. W. Shor, Additivity of the classical capacity of
entanglement-breaking quantum channels, 
J. Math. Phys. \textbf{43}  (2002), 4334-4340; 
arXiv:quant-ph/0201149.

\bibitem{Sh-2} M. E. Shirokov, 
The Convex Closure of the Output Entropy of
Infinite Dimensional Channels and the Additivity Problem, Russian
Mathematical Surveys \textbf{61} (2006), 1186-1188;
arXiv:quant-ph/0608090.

\bibitem{hs} A. S. Holevo and M. E. Shirokov, 
On Shor's channel extension and constrained channels, 
Commun. Math. Phys. \textbf{249} (2004), 417-430; arXiv:quant-ph/0306196.

\bibitem{h1} A. S. Holevo, On complementary channels and the additivity
problem, Probab. Theory and Appl. \textbf{51} (2006),  133-134;
arXiv:quant-ph/0509101.

\bibitem{mkr} C. King, K. Matsumoto, M. Nathanson, and M. B. Ruskai,
Properties of conjugate channels with applications to additivity and
multiplicativity, Markov Process and Related Fields \textbf{13} (2007),
391-423;  arXiv:quant-ph/0509126.

\bibitem{QCMC}
R. Garc{\'i}a-Patr\'on and N. Cerf, QCMC2012 talk (2012).

\bibitem{WATER}
A. S. Holevo, M. Sohma, and O. Hirota, 
Capacity of quantum Gaussian channels, 
Phys. Rev. A 
{\bf 59} (1999), 1820.

\bibitem{BROAD}
V. Giovannetti, S. Lloyd, L. Maccone, and P. W. Shor,
Broadband channel capacities,
Phys. Rev. A {\bf 68} (2003), 062323.

\bibitem{JOA}
J. Sch\"{a}fer, D. Daems,  E. Karpov, and N. J. Cerf,
Capacity of a bosonic memory channel with Gauss-Markov noise,
Phys. Rev. A {\bf 80} (2009),  062313.

\bibitem{OLEG}
O. V. Pilyavets, C. Lupo, and S. Mancini,
Methods for Estimating Capacities and Rates of Gaussian Quantum Channels,
IEEE Trans. Inf. Theory, {\bf 58} (2012), 6126-6164.



\end{thebibliography}
\end{document}